\documentclass[submission,copyright,creativecommons]{eptcs}
\providecommand{\event}{NCL 2026} 
\usepackage{array}
\usepackage{iftex}
\usepackage{amsthm}
\usepackage{amsmath}
\usepackage{amsfonts}
\usepackage{amssymb}
\usepackage{enumitem}
\usepackage{stix}
\usepackage{tikz}
\usepackage{capt-of}
\usepackage{hyperref}

\newtheorem{theorem}{Theorem}[section]

\newtheorem{fact}[theorem]{Fact}
\newtheorem{lemat}[theorem]{Lemma}
\newtheorem{wn}[theorem]{Corollary}

\ifpdf
  \usepackage{underscore}         
  \usepackage[T1]{fontenc}        
\else
  \usepackage{breakurl}           
\fi

\title{On Some Logics Intermediate Between S4.2 and S4.3 Defined by the Notion of Independence}

\def\titlerunning{Logics Between S4.2 and S4.3}
\def\authorrunning{M. A. Hałapacz}
\begin{document}
\author{Maciej A. Hałapacz
\institute{Department of Logic and Methodology of Sciences\\ University of Wrocław\\ Wrocław, Poland}\email{maciej.halapacz@uwr.edu.pl}}

\maketitle

\begin{abstract}
Hamkins and L\"{o}we described modal logic S4.2 through the notion of sets of independent buttons --- statements that are possibly necessary --- and switches --- necessarily contingent statements. In this article we investigate logics intermediate between S4.2 and S4.3 that are defined through a different notion of independence. We provide an axiomatic description of such logics, state some basic facts and prove various completeness and decidability results. 
\end{abstract}

\section{Introduction}

First-order structures can be investigated through the lens of modal logic: we investigate a propositional modal logic of the frame $\langle C, R \rangle$, where $C$ is a given class of first-order models of a certain type and $R$ is a relation on those structures. Since this is a frame for a propositional modal logic, valuations should be understood accordingly: let $t \colon Var \to Fm_L$ be a translation from the propositional to first-order language. For all such translations $t$ we consider those valuations $v_t$ such that $v_t(p) =\{ M \in C : M \models t(p)\}$. Various first-order classes of structures have been studied in such manner: models of set theory \cite{HL2008}, abelian groups \cite{BBL2021}, graphs \cite{WH2024} and lattices \cite{ja} under embeddings, as well as general theory of such frames with submodel relation \cite{SS}. Many of those logics turn out to be well-known modal logics such as $S4$, $S4.2$, $S4.2.1$, $S4.3$, $S5$. Those logics can be defined through the notions of \emph{buttons}, possibly necessary statements, and \emph{switches}, necessarily possible statements, as well as the notion of the families of those statements being \emph{independent}. In \cite{HL2008} the authors pose the question whether a given model of set theory under a forcing extension relation can yield any of the intermediate modal logics between S4.2 and S4.3, defined through the finite boundaries on the families of independent buttons. In this paper, we investigate the notion of finite boundaries on the size of buttons, point out some ambiguities within this proposed characterization and research new logics, defined within the same framework: we provide various completeness, as well as decidability results. 

On the other hand, our results are interesting from purely logical point of view, defining new classes of logics intermediate between $S4.2 $ and $S4.3$, similarly to Hughes's research in \cite{HU}.

\section{Modal logic}

Let $Var$ be a countable set of propositional variables. Then we take the set $\mathtt{P} = Var \cup \{\square, \land, \neg \}$ to be a language of propositional modal logic, where $\square, \neg$ are unary operators and $\land$ is binary. The set of all formulas of $\mathtt{P}$, $Fm_\mathtt{P} $ is defined in a standard way. Other operators ($\lozenge, \lor, \rightarrow, \leftrightarrow$) are defined in a standard way as well.

In the above vocabulary, we define a set $\Lambda \subseteq Fm_\mathtt{P}$ to be a \emph{modal logic} iff it is closed under Modus Ponens and uniform substitutions. $\Lambda$ is called a \emph{normal} modal logic if it contains the axiom $K$: \newline \mbox{$\square(p \rightarrow q) \rightarrow (\square p \rightarrow \square q)$}, and is closed under the necessitation rule: $\frac{\vdash \ \alpha}{\vdash \ \square \alpha}$.

Consider the following formulas:

\begin{align*}
    T:& \ \square p \rightarrow p, \\
    4:& \ \square p \rightarrow \square \square p, \\
    .2:& \ \lozenge \square p \rightarrow \square \lozenge p, \\
    .3:& \ \square(\square p \rightarrow q) \lor \square (\square q \rightarrow p).
\end{align*}
The modal logic $S4.2$ is the smallest normal modal logic containing formulas $K$, $T$, $4$, $.2$. Modal logic $S4.3$ is the smallest normal modal logic containing all formulas of $S4.2$ and the formula $.3$.

Hamkins and L\"{o}we \cite{HL2008} proved that the logic S4.2 is consistent with arbitrarily large finite families of independent buttons and switches. Let us clarify what that means.

Let $\mathfrak{M} = (M, R, v)$ be a Kripke model of a logic extending S4.2. A button in a possible world $m \in M$ is a sentence $\beta \in Fm_\mathtt{P}$ such that $m \Vdash \lozenge \square \beta$. A switch in a world $w \in M$, on the other hand, is a sentence $\sigma \in Fm_\mathtt{P}$ such that $w \Vdash \square (\lozenge \sigma \land \lozenge \neg \sigma )$. If $m \Vdash \lozenge \square \beta$ and $m \Vdash \square \beta$, we say that the button $\beta$ is \emph{pushed} (otherwise it is \emph{unpushed}). If $m \Vdash \square (\lozenge \sigma \land \lozenge \neg \sigma )$ and $m \Vdash \sigma$ we say that a switch $\sigma$ is \emph{on} (otherwise it is \emph{off}). A finite set $X$ of sentences that are buttons and switches is said to be \emph{independent} in a world $w$ iff for each set of buttons $B \subseteq X$ and each set of switches $S \subseteq X$, $w$ validates the following formula: 
\begin{equation}\label{butswi}
\square \lozenge (\bigwedge \square \beta \land \bigwedge \neg \square \beta' \land \bigwedge \sigma \land \bigwedge \neg \sigma'),\end{equation}
where $\sigma \in S$, $\sigma' \in X \setminus S$ are switches and $\beta \in B'$, $\beta' \in X \setminus B'$ are buttons in $w$, $B' = B \cup P$ and $P$ is the set of all buttons pushed in the evaluated world ($P = \varnothing \text{ in } w$). For each finite family of independent buttons and switches, there exists a model $\mathfrak{M}$ of S4.2 where a certain world $w \in M$ satisfies the formula \ref{butswi}. 

By contrast, the cardinality of families of independent buttons is at most 1 under the logic $S4.3$. Indeed, if we take any two buttons, the axiom $.3$ guarantees that pushing one of them also pushes the other --- suppose that $w \Vdash \lozenge \square p, \lozenge \square q$. That means there are some accessible worlds $u, v$ such that $u \Vdash \square p$ and $v \Vdash \square q$. But frames for $S4.3$ are connected:  $uRv$ or $vRu$, so $u \Vdash \square q$ or $v \Vdash \square p$. 

$S4.3$ can be characterized instead through the notion of a \emph{ratchet}. A ratchet in a possible world $m$ is a finite set of buttons $\{\beta_1, \dots, \beta_n\}$, such that for any $1 \leqslant i \leqslant n$ we have $m \Vdash \lozenge (\square\beta_i \land \neg \beta_j)$ for $n \geqslant j>i$ and $m \Vdash \square\beta_i \rightarrow \beta_k$, for $k \leqslant i$. 


This leads to the question: what are the modal logics that limit the cardinality of sets of independent buttons to $n$? The remainder of this section is devoted to answering this problem. However, this issue is not as straightforward as it may seem.

Note that to answer this question we need not mention the independence of families of switches with the families of buttons that are of interest. For this reason, we will omit referencing the switches in the rest of the paper. However, one can easily observe that the following results allow for the independence of arbitrarily large families of switches together with any notion of independence stated below.

Let us say that a logic $L$ has the property $B_n$ if $n$ is the greatest cardinality of sets of independent buttons allowed by the $L$. 
Let us first focus on the case where $n=2$. Then, we will generalize the results to an arbitrary natural number $n$.

Let us state the following fact:

\begin{fact}
    For any $n>1$ there are many modal logics with the property $B_n$.
\end{fact} \noindent
To observe this, consider the following formulas:
\begin{align*}
&A_2:\quad
\square[\square p \rightarrow
    (\square(\square q \rightarrow r)
    \lor \square(\square r \rightarrow q))] \lor{}\\
&\phantom{A_2:\quad}
\square[\square q \rightarrow
    (\square(\square p \rightarrow r)
    \lor \square(\square r \rightarrow p))] \lor{}\\
&\phantom{A_2:\quad}
\square[\square r \rightarrow
    (\square(\square p \rightarrow q)
    \lor \square(\square q \rightarrow p))],\\[6pt]
&I_2:\quad
\square(\square p \rightarrow q \lor r)
\lor \square(\square q \rightarrow p \lor r)
\lor \square(\square r \rightarrow p \lor q).
\end{align*}

Both the logic $S4.2 + A_2$ and $S4.2 +I_2$ are consistent with the families of two independent buttons. The frame based on a 4-element boolean lattice is a frame with both $A_2$ and $I_2$ valid, and there is a valuation $v$ where $v(p)$ is one ultrafilter of the lattice, while $v(q)$ is the other --- $p$ and $q$ are the buttons in any such model.

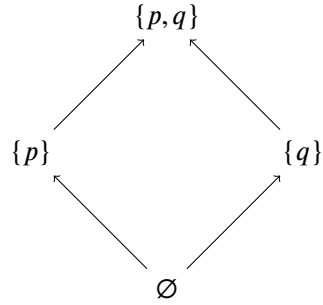
\begin{figure}[ht]
\centering
\begin{tikzpicture}[scale=1.2]
    \node (0) at (0,0) {$\varnothing$};
    \node (p) at (-1.5,1.5) {$\{p\}$};
    \node (q) at (1.5,1.5) {$\{q\}$};
    \node (1) at (0,3) {$\{p,q\}$};

    \draw[->] (0) -- (p);
    \draw[->] (0) -- (q);
    \draw[->] (p) -- (1);
    \draw[->] (q) -- (1);
\end{tikzpicture}
\caption{A model satisfying two independent buttons}
\label{fig:boole4}
\end{figure}

However, there are no families of independent buttons of size $\geqslant 3$ in either logic. Obviously, both $A_2$ and $I_2$ exclude this possibility, but $I_2$ is stronger. Let us state some facts about those logics so that this difference is clearer.

Let us use $wR,w_1, w_2, \dots w_n$ as an abbreviation for $wRw_1, wRw_2, \dots wRw_n$. Let us denote by $|w|$ the set of all worlds accessible from $w$. Let us say that $|w|$ is connected if and only if for each $u, v \in |w|$, $uRv$ or $vRu$. Consider the following condition: for all possible worlds $w_0, w_1, w_2, w_3$ we have \begin{center}
   $R_2$: if $w_0Rw_1, w_2,w_3$ and none of the $w_1, w_2, w_3$ are accessible from each other, then $|w_i|$ is connected, for some $i \in \{1, 2, 3\}$. \end{center}
   A frame $F$ is \emph{directed} iff for all $w, u, v \in F$ if $w Ru, v$ then there exists $z$ such that $uRz$ and $vRz$.
\begin{theorem}\label{c1a2}
    The logic $S4.2 + A_2$ is complete with respect to the class of reflexive, transitive, directed Kripke frames that satisfy the condition $R_2$.
\end{theorem}
\begin{proof}
    For soundness, it suffices to check that the axiom $A_2$ is valid in such frames. Let $\mathbb{F} = (F, R)$ be such a frame. Suppose \emph{a contrario} that $w_0 \in F$ is a world such that it falsifies $A_2$. It follows that $(\exists w_1, w_2, w_3 \in F) (w_0Rw_1, w_2, w_3$ and $w_1 \Vdash \square p, w_2 \Vdash 
    \square q, w_3 \Vdash \square r)$, and $w_1 \not\Vdash \square (\square q \rightarrow r) \lor \square (\square r \rightarrow q)$, and $w_2, w_3$ each falsify the appropriate consequent of the implications of $A_2$. Without loss of generality, assume that in $\mathbb{F}$, $|w_1|$ is connected. Since $w_1 \Vdash \square p$ and $w_2 \Vdash \square q$, $u \Vdash \square (p \land q)$ in any $u$ such that $w_1, w_2Ru$. There is such $u$, because the frame is directed. But since directedness also forces that $w_1$ and $w_3$ have a common extension $u'$, then $u' \Vdash \square r$. Because $u, u' \in |w_1|$, we have $uRu'$ or $u'Ru$. Without loss of generality we can assume that $u'Ru$ and then we have that $w_1 \Vdash \square p \rightarrow \square(\square r \rightarrow q)$ --- a contradiction, that concludes this part of the proof.

For completeness, we will use canonical models. We only have to check that the canonical model $\mathfrak{A}_2$ for $S4.2 + A_2$ satisfies condition $R_2$. Suppose indirectly that $w_0 R w_1$ $\&$  $w_0 R w_2$ $\&$  $w_0 R w_3$ and $w_1, w_2, w_3$ are not accessible from each other, but each of $w_1, w_2, w_3$ has access to two worlds that are not accessible from each other: $w_1Rs, t$, $w_2Ru, w$, $w_3Rv, z$. But if it is so, then we have $\square \alpha, \neg \beta \in s$, $\square \beta , \neg \alpha \in t$, $\square \alpha' \neg \beta' \in u$,  $\square \beta', \neg \alpha' \in w$, $\square \alpha'',\neg \beta'' \in v$, $\square \beta'', \neg \alpha'' \in z$. But then (as long as $\{s, t\} \neq \{w, u\} \neq \{v, z\}$, but if such an identity would hold, then $R_2$ would as well) $w_0 \not\Vdash A_2$, for substitution $(\alpha\lor\alpha')/p, (\beta \lor \alpha'')/q, (\beta' \lor \beta'')/r$ --- a contradiction.
\end{proof}
Now that completeness has been established, it is natural to ask if one can limit oneself in investigations of $S4.2+A_2$ to finite models --- or, in simpler worlds, if the investigated logic has the finite model property. It turns out that it does.
\begin{theorem}\label{fmpa2}
    The logic $S4.2 + A_2$ has the finite model property.
\end{theorem}
Before turning our attention to the proof, some preliminary remarks are in order.

There is a canonical model for $S4.2+A_2$, call it $\mathfrak{M}$. Let $X \subseteq Fm_{\mathtt{P}}$ be closed under subformulas, and let $X' = \{\neg \alpha \in Fm_\mathtt{P}: \alpha \in X\}$ Consider the following relation on the members of $\mathfrak{M}$'s domain: $w \sim u \iff (\forall \alpha \in X') (w \Vdash \alpha \iff u \Vdash \alpha)$. We denote an equivalence class of a world $w$ with respect to $\sim$ by $[w]$. We consider the following model $\mathfrak{M}_{X'} = \langle M_{X'}, R_{X'}, v_{X'} \rangle$:

\begin{enumerate}
  
    \item $M_{X'} = \{[w] : w \in \mathfrak{M} \}$ 
    \item $[w]R_{X'}[u] \iff (\forall \alpha \in X')( \lozenge \alpha \in u \Rightarrow \lozenge \alpha \in w)$
    \item $v_{X'}(p) = \{[w]: p \in w\}$, for all $p \in Var \cap X$. 
\end{enumerate}
  \begin{lemat}
      $\mathfrak{M}_{X'}$ is a filtration of $\mathfrak{M}$ through $X'$.
  \end{lemat}
  \begin{proof}
      Since cases 1. and 3. are in line with the standard filtration definition, it suffices to check that $R_{X'}$ validates the following conditions:
      \begin{enumerate}[label=(\arabic*)]
          \item $wRu \Rightarrow [w]R_{X'}[u]$,
          \item $[w]R_{X'}[u] \Rightarrow (\forall \alpha \in X')(\alpha \in u \Rightarrow \lozenge \alpha \in w)$.
      \end{enumerate}
      They are however immediate from reflexivity and transitivity of $\mathfrak{M}$ (note that this proves that such a filtration is only a filtration for reflexive and transitive models, but our canonical model is like that).
  \end{proof}

\begin{proof}[Proof of theorem \ref{fmpa2}]
    Since $S4.2 + A_2$ is an extension of the logic \textbf{K}, the filtration theorem occurs. We just need to check that the filtered model $\mathfrak{M}_{X'}$ is a model of $S4.2 + A_2$. Checking the axioms characteristic for S4.2 is easy, we will only check the validity of $A_2$ here. For that, we need to check that it preserves connectedness.
    
    Suppose that for some world $w \in \mathfrak{M}$, $|w|$ is connected. Let us assume by contradiction, that there are some $u, v \in |w|$ such that $\neg [u] R [v]$ and $\neg [v] R [u]$. Without loss of generality, let us assume, that $\alpha, \beta \in X'$. The world $w$ satisfies the connectedness axiom and is a maximally consistent world, so it can be assumed without loss of generality that $\square (\square \alpha \rightarrow \beta)$. It follows that $\square \alpha \rightarrow \beta \in u, v$, so $\neg \square \alpha \lor \beta \in u,v$. (i) Suppose that $\beta \in u$. It follows that $\lozenge \beta \in u$, hence $\neg \lozenge \beta \in v$. This means that $\neg \square \alpha \in v$, so $\lozenge \neg \alpha \in v$, and it follows from the lack of accessibility of $v$ and $u$ that $\square \neg \alpha \in v$. But $\square (\square \neg \alpha \rightarrow \beta) \lor \square (\square \beta \rightarrow\neg \alpha) \in w$, because of the fact that $\neg \alpha \in X'$. However, $\square \neg \alpha \rightarrow \beta \notin v$ and $\square \beta \rightarrow \neg \alpha \notin u$. (ii) Supose that $\neg \square \alpha \in u$. Then similar argument leads to $\square \neg \alpha \rightarrow \beta \notin u$ and $\square\beta \rightarrow \neg \alpha \notin v$, and this contradiction concludes the proof.
    \end{proof}



\begin{wn}
The logic $S4.2 + A_2$ is decidable.
\end{wn}
As mentioned above, the axiom $A_2$ directly violates the possibility of consistency for the logic $S4.2 +A_2$ with a set of 3 or more independent buttons. However, one can opt for a different notion of independence. Let us call a finite set of buttons $X$ \emph{weakly independent} in $w$ iff for every $A \subseteq X$, $w \Vdash \lozenge (\bigwedge_{\alpha \in A} \alpha \land \bigwedge_{\beta \in X \setminus A} \beta)$. 

\begin{fact}
    Logic $S4.2 +A_2$ is consistent with arbitrarily large families of weakly independent buttons.
\end{fact}
\begin{proof}
    For such a set $X$ of size $n$ take a frame isomorphic to a lattice $M_{2^n}$ (least element, greatest element, and $2^n$ pairwise incomparable elements in the middle). Let the valuation assign to different "middle" worlds different subsets of $X$ (note that it is not the most efficient way of constructing such frames --- compare with figure 2). This is of course a $S4.2+A_2$-frame, since all worlds above the initial world satisfy connectedness.
\end{proof}
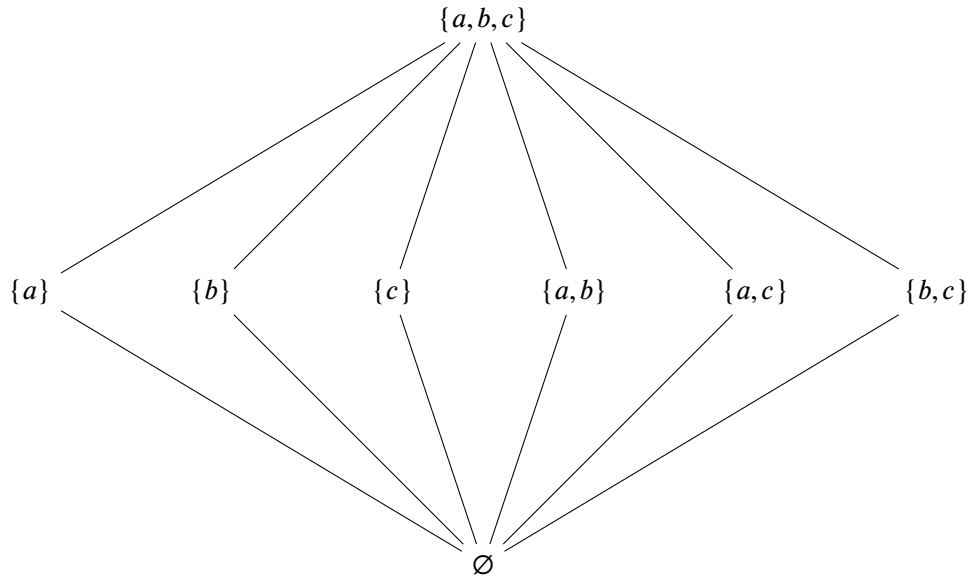
\begin{figure}[ht]
\centering
\begin{tikzpicture}[scale=1.2]

\node (0) at (0,0) {$\varnothing$};

\node (a)  at (-5,3) {$\{a\}$};
\node (b)  at (-3,3) {$\{b\}$};
\node (c)  at (-1,3) {$\{c\}$};
\node (ab) at ( 1,3) {$\{a,b\}$};
\node (ac) at ( 3,3) {$\{a,c\}$};
\node (bc) at ( 5,3) {$\{b,c\}$};

\node (1) at (0,6) {$\{a,b,c\}$};

\draw (0) -- (a);
\draw (0) -- (b);
\draw (0) -- (c);
\draw (0) -- (ab);
\draw (0) -- (ac);
\draw (0) -- (bc);

\draw (a)  -- (1);
\draw (b)  -- (1);
\draw (c)  -- (1);
\draw (ab) -- (1);
\draw (ac) -- (1);
\draw (bc) -- (1);

\end{tikzpicture}
\caption{A model satisfying 3 weakly independent buttons}
\label{fig:M6}
\end{figure}
More different notions of independence can be constructed, and hence different logics are consistent with arbitrary sets of differently-independent buttons. This is of course consistent with satisfying $B_2$ (or any $B_n$, for that matter).

Now, let us move on to logic $S4.2+I_2$, which differs from the former logic in this aspect. Let us state some facts about this logic as well. Consider the following condition:
\begin{center}
    $R'_2$: If $w_0Rw_1, w_2, w_3$ then $w_iRw_j$, for some distinct $i, j \in \{1, 2, 3\}$.
\end{center}
This can be stated equivalently as:
\begin{center}
    $R'_2$: There are no $R$-antichains of cardinality $>2$ with a common predecessor.
\end{center}
\begin{theorem}\label{cli2}
    $S4.2 +I_2$ is complete with respect to the class of Kripke frames that are reflexive, transitive, directed and satisfy condition $R'_2$.
\end{theorem}
\begin{proof}
For soundness suppose that $\mathbb{F}$ is a frame satisfying the above conditions and $w \not\Vdash I_2$. That means there are worlds $u, v, s$ such that $wRu, v, s$ and $u \Vdash \square p \land \neg q \land \neg r$, $v \Vdash \square q \land \neg r \land \neg p$ and $s \Vdash \square r \land \neg p \land \neg q$. However, one of them accesses at least one of the other, and so $\square \alpha$ is valid where $\neg \alpha$ is valid, for some $\alpha \in \{p, q, r\}$, which is a contradiction.

For completeness we will show that the canonical model satisfies $R'_2$ (other conditions are obvious). Suppose that there are worlds $w_0, w_1, w_2, w_3$ and $w_0Rw_1, w_2, w_3$. Would $w_1, w_2, w_3$ be mutually inaccessible, there would be formulas $\alpha, \beta, \gamma$ such that $w_1 \Vdash \square \alpha$ yet $w_2, w_3 \not\Vdash \alpha$, $w_2 \Vdash \square \beta$, but $w_1, w_3 \not\Vdash \beta$ and $w_3 \Vdash \square \gamma$ while $w_1, w_2 \not\Vdash \gamma$. But without loss of generality because $w_0$ is a maximally consistent set, $w_0 \Vdash \square(\square \alpha \rightarrow \gamma \lor \beta)$, so either $w_1 \Vdash \beta$, or $w_1 \Vdash \gamma$, a contradiction. This concludes the proof.
\end{proof}
Similarly as in the logic $S4.2 +A_2$, completeness theorem can be limited to finite models. 

\begin{theorem}\label{fmpi2}
    $S4.2+I_2$ has the finite model property.
\end{theorem}
\begin{proof}
    In order to prove this it suffices to use a standard filtration (through arbitrary set $X$ closed under subformulas), namely the \emph{largest} filtration, where $[w]R_X[u] \iff (\forall \alpha \in X)(\square \alpha \in w \Rightarrow \alpha \in u)$. This filtration preserves reflexivity, transitivity and directedness. Moreover, it preserves $R'_2$, as from the fact that in filtrations it is the case that $wRu \Rightarrow[w]R[u]$. 
\end{proof}
\begin{wn}
    The logic $S4.2+I_2$ is decidable.
\end{wn}

This logic differs from the previous one, since there can be no rational notion of arbitrarily large sets of in-some-sense-independent buttons:
\begin{fact}
    If $B = \{\beta_1, \dots, \beta _n\}$ is a set of buttons, then if $L, P \subseteq B$ are nonempty and such that $P \cap L = \varnothing $, then there is a model of $S4.2+I_2$ satisfying $\lozenge (\bigwedge_{\alpha \in P}\alpha \land \bigwedge_{\gamma \in L}\neg \gamma) \land \lozenge (\bigwedge_{\alpha' \in L}\alpha' \land \bigwedge_{\gamma' \in P}\neg \gamma')$. However, for any $n$ distinct subsets of $B$, for $n>2$, there is no such model.
\end{fact}
\begin{proof}
    Observe that one can satisfy such formula in a model built on a frame like on a figure 1. The limit on the width of the model forbids that such a composition be possible with three or more subsets of $B$.
\end{proof}

These results can be easily generalized. For any natural number $n \geqslant 4$, let $X$ be a set of all its subsets of size $n-2$.  Consider the following axioms:

\begin{align*}
 &A_n:  \bigvee_{Y \in X} \square ( \bigwedge_{k \in Y, i, j \in n \setminus Y} \square p_k \rightarrow (\square (\square p_i \rightarrow p_j) \lor \square (\square p_j \rightarrow p_i))), \\
    &I_n: \bigvee_{i \in n} \square (\square p_i \rightarrow \bigvee_{j \in n \setminus \{i\}} p_j);
\end{align*}
and the following conditions:
\begin{align*}
    &R_n: \text{for any worlds }w_0, \dots w_{n+1}\text{ if }w0Rw_1, \dots w_{n+1}\text{ then }|w_i|\text{ is connected, for some }i \in \{1, \dots, n\}, \\
    &R'_n\text{: There are no R-antichains of cardinality >n with a common predecessor.}
\end{align*}
\begin{theorem}
    Each $S4.2+A_n$ is complete with respect to the class of all frames satisfying $R_n$.
\end{theorem}
\begin{theorem}
    Each $S4.2+I_n$ is complete with respect to the class of all frames satisfying $R'_n$.
\end{theorem}
\begin{theorem}
    Each $S4.2 + A_n$ has the finite model property, and so is decidable. 
\end{theorem}
\begin{theorem}
    Each $S4.2 + I_n$ has the finite model property, and so is decidable.
\end{theorem}
The proofs are simple generalizations of the proofs of theorems \ref{c1a2}, \ref{cli2}, \ref{fmpa2}, \ref{fmpi2} respectively.
\begin{fact}
    For any $n \geqslant 2$ the following hold: \begin{equation}\label{1}
        B_n(S4.2+A_n), B_n(S4.2+I_n)
    \end{equation}\begin{equation}
        S4.2 +A_n \supset S4.2+A_{n+1},
    \end{equation} \begin{equation}
        S4.2 +I_n \supset S4.2 + I_{n+1},
    \end{equation}\begin{equation}\label{przedostatni}
        S4.3 \supset S4.2 +I_n \supset S4.2 +A_n \supset S4.2,
    \end{equation}\begin{equation}\label{ostatni}
        \bigcap_{n \geqslant 2}S4.2 + A_n = \bigcap_{n \geqslant2} S4.2 +I_n = S4.2.
    \end{equation}
\end{fact}
\begin{proof}
    (\ref{1})-(\ref{przedostatni}) are trivial. For (\ref{ostatni}) observe that a boolean lattice $\mathbb{B}_n$ is a frame for $S4.2+I_{n+1}$, and so for $S4.2+A_{n+1}$ as well.
\end{proof}
\section{Summary}

The logics that limit the size of independent buttons can be widely different. It seems that the strongest of them are those of the form $S4.2+I_n$, for $n \geqslant 1$. However, there are other ways of defining the notion of independence, which can allow us to discover new logics consistent with arbitrarily large families of buttons independent in this new sense. 

Further research could focus on defining the logics (or the appropriate independence notion) so that there is a correspondence between a class of frames within which one can easily navigate due to the buttons and the logic of this class. 

\nocite{*}
\bibliographystyle{eptcs}
\bibliography{generic}
\end{document}